\documentclass[11pt]{article}

\usepackage{amsmath}
\usepackage{amssymb}
\usepackage{amsthm}
\usepackage[paper=letterpaper,margin=1in]{geometry}
\usepackage[ruled,vlined,linesnumbered]{algorithm2e}
\usepackage{forloop}
\usepackage{bm}

\usepackage{xspace}

\usepackage[l2tabu, orthodox]{nag}

\usepackage{hyperref}

\usepackage{nth}

\usepackage{array}

\usepackage{paralist}

\usepackage{flafter}

\usepackage{tikz}
\usetikzlibrary{calc}
\usetikzlibrary{matrix}
\usetikzlibrary{shapes}
\usetikzlibrary{decorations.pathmorphing}
\usetikzlibrary{decorations.pathreplacing}

\newtheorem{theorem}{Theorem}[section]
\newtheorem{lemma}[theorem]{Lemma}

\newtheorem{corollary}[theorem]{Corollary}

\newtheorem{reduction}{Reduction}
\newtheorem{claim}[theorem]{Claim}

\makeatletter
\newtheorem*{rep@theorem}{\rep@title}
\newcommand{\newreptheorem}[2]{%
\newenvironment{rep#1}[1]{%
 \def\rep@title{#2 \ref{##1}}%
 \begin{rep@theorem}}%
 {\end{rep@theorem}}}
\makeatother

\newreptheorem{reduction}{Reduction}

\colorlet{darkgreen}{green!50!black}

\newcommand{\defcal}[1]{\expandafter\newcommand\csname c#1\endcsname{{\mathcal{#1}}}}
\newcounter{ct}
\forLoop{1}{26}{ct}{
    \edef\letter{\Alph{ct}}
    \expandafter\defcal\letter
}

\DeclareMathOperator{\spn}{span}

\DeclareMathOperator{\even}{even}
\DeclareMathOperator{\odd}{odd}

\renewcommand{\vec}[1]{{\bm{#1}}}
\newcommand{\ie}{i.e.,\xspace}
\newcommand{\eg}{e.g.,\xspace}

\newcommand{\E}{{\mathbb{E}}}
\newcommand{\msp}{{\ensuremath{\mathsf{MSP}}}}
\newcommand{\sbmsp}{{\ensuremath{\mathsf{SB}$-$\mathsf{MSP}}}}
\newcommand{\OPT}{{\ensuremath{\mathrm{OPT}}}\xspace}
\newcommand{\var}{{\ensuremath{\mathsf{Var}}}}

\newcommand\rank{r}

\newcommand{\hide}[1]{}

\title{\textbf{A Simple 
$O(\log\log(\mathrm{rank}))$-Competitive
Algorithm for the Matroid Secretary Problem}}
\author{Moran Feldman\thanks{School of Computer and Communication Sciences, EPFL. 
Email:
\href{mailto:moran.feldman@epfl.ch}{moran.feldman@epfl.ch}. Supported by ERC Starting Grant 335288-OptApprox.},
Ola Svensson\thanks{School of Computer and Communication Sciences, EPFL.
Email:
\href{mailto:ola.svensson@epfl.ch}{ola.svensson@epfl.ch}. Supported by ERC Starting Grant 335288-OptApprox.},
Rico Zenklusen\thanks{Department of Mathematics, ETH Zurich
and Department of Applied Mathematics and Statistics,
Johns Hopkins University.
Email:
\href{mailto:ricoz@math.ethz.ch}{ricoz@math.ethz.ch}.
Partially supported by NSF grant CCF-1115849.}}

\begin{document}

\maketitle
\thispagestyle{empty}

\begin{abstract}

%
Only recently progress has been made
in obtaining $o(\log(\mathrm{rank}))$-competitive
algorithms for the matroid secretary problem.
More precisely, Chakraborty and Lachish (2012) presented
a $O(\sqrt{\log(\mathrm{rank})})$-competitive procedure,
and Lachish (2014) later presented a
$O(\log\log(\mathrm{rank}))$-competitive algorithm.
Both these algorithms and their analyses
are very involved, which is also
reflected in the extremely high constants
in their competitive ratios.

Using different tools, we present 
a considerably simpler $O(\log\log(\mathrm{rank}))$-competitive
algorithm for the matroid secretary problem.
Our algorithm can be interpreted as a distribution over a
simple type of matroid secretary algorithms which are
easy to analyze.
Due to the simplicity of our procedure, we are also able
to vastly improve on the hidden constant in the competitive
ratio.

%

\end{abstract}


\newpage

\setcounter{page}{1}

\section{Introduction}

The secretary problem is a classical online 
selection problem, whose origins remain
unclear~\cite{dynkin_1963_optimum,%
ferguson_1989_who,%
gardner_1960_mathematical,gardner_1960b_mathematical,%
lindley_1961_dynamic}.
In its original form, the task is to select the
best out of a set $N$ of $n$ secretaries (also called
\emph{elements} or \emph{items}). 
Secretaries appear one by one in a uniformly random order.
Whenever a secretary appears, he can be compared against
all previously appeared secretaries. Then, the algorithm
has to decide, before the arrival of the next secretary, 
whether to select the current secretary or not.
A well-known classical algorithm~\cite{dynkin_1963_optimum} 
selects the best secretary
with probability at least $1/e$, and this is known
to be asymptotically optimal.

Recently, there has been an increased interest in
variations of the secretary problem. Such variants have numerous applications in mechanism design for settings involving the selling of a good to agents arriving online.
In this context, the secretaries correspond to the
agents, and their values are the prices they are
willing to pay for the available goods
(see~\cite{azar_2014_prophet,%
babaioff_2008_online,%
babaioff_2007_matroids,%
kleinberg_2005_multiple-choice} and the references therein).
These applications naturally lead to generalized
secretary problems, where more than one element
can be selected.
In such problems, one typically assumes that each
element $e\in N$ reveals a positive weight $w(e)$, and the
goal is to select a maximum weight set of elements
subject to some constraints.
Most of these problems preserve the uniformly random arrival order of the elements, but allow adversarial assignment of weights. Like in the original problem, whenever an element appears, the algorithm must decide immediately, and irrevocably, whether to select it.

The arguably most canonical generalization
of the secretary problem was introduced by
Kleinberg~\cite{kleinberg_2005_multiple-choice},
who considered the problem of
selecting $k$ out of $n=|N|$ secretaries.
%
%
%
However, many applications require more general
constraints, and thus, interest arose in finding
relevant and general constraint classes for which
strong online algorithms exist.
This led to the introduction of the matroid
secretary problem~\cite{babaioff_2007_matroids},
where the underlying constraint
set is assumed to be a matroid $M=(N,\mathcal{I})$
defined over the set $N$ of all $n$
items.\footnote{
A matroid $M=(N,\mathcal{I})$ is a tuple consisting
of a finite ground set $N$, and a nonempty
family $\mathcal{I}\subseteq 2^N$ of subsets of
the ground set, called \emph{independent}
sets, which satisfy:
(i) $I\subseteq J \in \mathcal{I} \Rightarrow I\in \mathcal{I}$,
and
(ii) $I,J\in \mathcal{I}, |I|>|J|$ $\Rightarrow$
$\exists e\in I\setminus J$ s.t.
$J\cup\{e\}\in \mathcal{I}$.
} 
Matroid constraints model many interesting settings,
and it was conjectured that there exists an algorithm
which is $O(1)$-competitive for any matroid
constraint~\cite{babaioff_2007_matroids}.
We recall that an algorithm is $c$-competitive
for some $c\geq 1$ if it returns 
an independent set $I\in \mathcal{I}$ whose expected
weight is at least $\frac{1}{c}w(\OPT)$, where $w(\OPT)$
is the weight of the offline optimum \OPT, i.e.,
the maximum weight independent set.\footnote{For
simplicity, we assume all weights are disjoint,
which implies the existence of a unique maximum weight
independent set (which is also a base of $M$, \ie its size is 
equal to the rank of $M$). This assumption is without loss of
generality since one can break ties between weights
arbitrarily.}
Motivated by the above conjecture, $O(1)$-competitive
algorithms have been obtained for
a wide variety of special classes of
matroids including graphic
matroids~\cite{babaioff_2007_matroids,korula_2009_algorithms}%
,
transversal matroids~\cite{babaioff_2007_matroids,%
dimitrov_2008_competitive,korula_2009_algorithms}%
,
co-graphic matroids~\cite{soto_2013_matroid}%
,
linear matroids with at most $k$ non-zero entries
per column~\cite{soto_2013_matroid}%
,
laminar matroids~\cite{im_2011_secretary,%
jaillet_2013_advances,%
ma_2013_simulated}%
,
regular matroids%
~\cite{dinitz_2013_matroid}, and
some types of decomposable matroids, including
max-flow min-cut matroids~\cite{dinitz_2013_matroid}%
.
However, progress on the general case has been
much slower. 
Since the introduction of the matroid secretary problem,
a simple $O(\log \rho)$-competitive algorithm was
known~\cite{babaioff_2007_matroids}, where $\rho$ is the
rank of the underlying matroid, i.e., the cardinality
of a maximum size independent set.
Improving on this bound has shown to be surprisingly
difficult.
So far, the only improvements on this bound are an
$O(\sqrt{\log \rho})$-competitive algorithm
by Chakraborty and Lachish~\cite{chakraborty_2012_improved},
and a very recent $O(\log\log(\rho))$-competitive procedure
by Lachish~\cite{lachish_2014_competitive-ratio_focs}.
Both algorithms use a careful bucketing of the ground
set and their competitive ratios are derived through
very involved analyses.
The complexity of the analyses of the above algorithms
is also reflected in the
hidden constant of the competitive ratio which is at
least $2^{64}$ for the $O(\sqrt{\log \rho})$-competitive
algorithm suggested in~\cite{chakraborty_2012_improved},
and at least $2^{2^{34}}$ for the 
$O(\log\log(\rho))$-competitive procedure
by Lachish~\cite{lachish_2014_competitive-ratio}.
%
%

In this paper we present a 
much simpler $O(\log\log \rho)$
procedure for the matroid secretary problem, which
also vastly improves the hidden constant
of the competitive ratio.
Our algorithm is 
\emph{order-oblivious}, which implies that it extends to
single-sample prophet inequalities as introduced by
Azar, Kleinberg and Weinberg~\cite{azar_2014_prophet}.
We expand on this connection below in
Section~\ref{subsec:orderObl}.

We would like to highlight that
like some previous matroid secretary algorithms,
all the information our algorithm needs to know upfront is the
size $n$ of the matroid.
During its execution, the algorithm only checks the independence of subsets of elements revealed so far.
%

\subsection{Single-sample prophet inequalities
and order-obliviousness}
\label{subsec:orderObl}

Prophet inequalities are a class of problems that is
closely related to secretary problems and has interesting applications in
mechanism design.
Unlike in secretary problems, in prophet inequality problems the weight of each element $e \in N$ is drawn from an element specific distribution $\mathcal{D}_e$ (the amount of knowledge the algorithm has on $\mathcal{D}_e$ varies according to the specific variant at hand). However, the order in which elements arrive is adversarial (rather than random), and might depend on the realization of the weights.
Azar, Kleinberg and Weinberg~\cite{azar_2014_prophet}
showed that interesting results can often be
obtained even if one only knows a single sample
from each distribution $\mathcal{D}_e$, which
is a setting they call \emph{single-sample
prophet inequalities}.

More precisely, they showed that any $c$-competitive
algorithm for the secretary problem can be transformed
into a $c$-competitive algorithm for single-sample
prophet inequalities, if the secretary algorithm
is \emph{order-oblivious}.
An order-oblivious procedure is one that consists of two phases: in the first phase the algorithm
specifies a (possibly random) number $m$ of the elements, and then observes a uniformly random
subset of $m$ elements without selecting any of them. The rest of the elements arrive in the second
phase, and the algorithm can select them. However, the competitive ratio of the algorithm must hold
for any order in which the elements of the second phase arrive. In other words, the elements
of the second phase might arrive in an adversarial order.
%
Thus, an order-oblivious algorithm uses only a small amount of the randomness of the arrival order, namely, whether each element appears among the first $m$ elements or not.



\subsection{Our results}

Our main result is the following.

%

\begin{theorem}\label{thm:mainThm}
There exists an order-oblivious
$2560[\log_2 \log_2 (4\rho) + 5]$-competitive
algorithm for the matroid secretary problem,
which only needs to know
the cardinality of the matroid upfront.
\end{theorem}

Our algorithm considerably improves on
the previous 
$o(\log \rho)$-competitive algorithms in terms
of simplicity, and we believe that a key contribution
of our work lies in the employed techniques and, arguably,
concise analysis.
These also lead
to a vastly reduced hidden constant
in the competitive ratio.
We recall that the hidden constant of Lachish's
$O(\log\log(\rho))$-competitive algorithm is
at least $2^{2^{34}}$.

The order-obliviousness of our algorithm
allows us to leverage the recent
reduction by Azar, Kleinberg
and Weinberg~\cite{azar_2014_prophet} to transform
our procedure into an algorithm for single-sample
prophet inequalities on matroids, leading to the following.
\begin{corollary}
There exists a
$2560[\log_2 \log_2 (4\rho) + 5]$-competitive
single-sample prophet inequality for any matroid.  
\end{corollary}



\subsection{Further related results}\label{subsec:furtherRes}

Some progress has been made in obtaining
$O(1)$-competitive algorithms for restricted variants
of the matroid secretary problem. In particular,
if a set of $n$ weights is assigned uniformly at random to the elements of the ground set, then
a $5.7187$-competitive algorithm can be obtained for
any matroid~\cite{soto_2011_matroid,soto_2013_matroid}.
%
Additionally, a $16(1-1/e)$-competitive algorithm
can be obtained
even with adversarial arrival order of the elements
as long as the weight assignment is still done at
random~\cite{oveisgharan_2013_variants,soto_2013_matroid}.
Furthermore, a $4$-competitive algorithm can
be obtained in the so-called \emph{free order model},
which assumes adversarial weight assignment, but allows
the algorithm to choose the order in which the elements
appear~\cite{azar_2014_prophet,jaillet_2013_advances}.

Variants of the matroid secretary problem involving nonlinear
objective functions, including 
submodular and convex objectives, were also considered~\cite{siddharth_2012_secretary,bateni_2013_submodular,
feldman_2011_improved, gupta_2010_constrained,
ma_2013_simulated}.

\subsection{A rough outline of our approach}

Our approach involves roughly three steps. The first step is a basic reduction that allows us to
assume a known upper bound $\tilde{\rho}$
on the rank $\rho$ of the matroid and that all weights are within a range $(W/(8\tilde{\rho}),W]$ for some known value $W$.
The second step is a simple secretary algorithm,
which we call the \emph{bucketing-based algorithm}. This algorithm gets
a partition (bucketing) of the secretaries, and produces a feasible solution whose quality depends on the input bucketing.
Our final algorithm simply picks a bucketing from an appropriately
chosen distribution, and then feeds it into the bucketing-based algorithm.
The three steps roughly correspond, in that order,
to Sections~\ref{sec:preliminaries}, \ref{sec:algo} and \ref{sec:full_algo}.

Our bucketing-based algorithm, which we introduce
formally in Section~\ref{sec:algo}, uses the bucketing it receives to define two collections of matroids $M_1, M_3, \dots, M_{2k - 1}$ and $M_2, M_4 \ldots M_{2k}$
with disjoint ground sets $N_1, N_2, \ldots, N_{2k} \subseteq N$,
having the following property:
if $I_i\subseteq N_i$ is independent in $M_i$ for $i \in \{1, 2, \ldots, 2k\}$, 
then $\bigcup_{i=1}^k I_{2i - 1} \in \mathcal{I}$ and $\bigcup_{i=1}^k I_{2i} \in \mathcal{I}$.
The bucketing-based algorithm chooses one of these collections at random, and then
simply selects greedily an independent set for each matroid $M_i$ in the chosen collection. The output of the algorithm is the union of the selected sets.
Similar decomposition ideas have been used previously
in the context of the matroid secretary problem
(see, e.g.,~\cite{jaillet_2013_advances,soto_2013_matroid}).
The main challenge lies in finding an elegant way to
analyze the behavior of such algorithms as a function of the decomposition, and then leveraging this analysis to find an appropriate probability distribution over the possible decompositions.
We achieve this goal for the bucketing-based
algorithm by presenting
lower bounds on the probability that elements get selected
by the algorithm.

\section{Preliminaries and basic reductions} \label{sec:preliminaries}

Let us formally state the Matroid Secretary Problem (\msp). An instance of {\msp} consists of a
matroid $M = (N, \cI)$ and a positive weight function $w : N \rightarrow \mathbb{R}_{>0}$. The
objective of an algorithm for {\msp} is to select a maximum weight independent set of
$M$. Initially, the algorithm knows the size of the ground set $n = |N|$, but has no other
information about either $M$ or $w$. Then, the elements of $N$ are revealed to the algorithm in a
uniformly random order. Each time an element $e \in N$ is revealed, the algorithm learns its weight
$w(e)$ and must decide immediately, and irrevocably, whether to select it. The algorithm also has
access to an independence oracle that,
given a subset $T \subseteq N$ of elements that \emph{already}
arrived, answers whether $T \in \cI$.

To simplify the exposition of our algorithm, we show how to apply it to a close variant of {\msp}
that we call \emph{Sample-Based Matroid Secretary Problem} (\sbmsp). {\sbmsp} shares the instance structure and
objective of {\msp}. However, the interaction of the algorithm with the instance is
different and it does not know the size $n$ of the matroid in advance. Initially, the algorithm must specify a (possibly random) sampling probability $p_s$.
The instance is then revealed in two
phases. In the first phase a random set $S \subseteq N$ containing every element $e \in N$ with
probability $p_s$ is revealed to the algorithm (along with the corresponding weights). The
algorithm, however, is not able to select any element of $S$. In the second phase, the elements of
$N \setminus S$ are revealed (together with their weights) in an adversarial  order that might depend
on the set $S$. Like in {\msp}, the algorithm  also has access to an independence oracle that given a
subset $T \subseteq N$ of elements that \emph{already} arrived answers whether $T \in \cI$.

\begin{reduction} \label{re:timed} 
  Any $\alpha$-competitive algorithm for {\sbmsp} can be transformed efficiently into an
  order-oblivious $\alpha$-competitive algorithm for {\msp}.
\end{reduction}
The above reduction follows from standard arguments, and we defer its formal proof to
Appendix~\ref{sec:redProofs}. Intuitively, an algorithm for \msp{} can be obtained from an algorithm
for \sbmsp{} (using sample probability $p_s$) as follows: let the sample set $S$ be roughly $p_s n$
of the first elements arriving according to the random permutation; and then proceed exactly in the
same manner. The resulting algorithm is order oblivious as the algorithm for \sbmsp{} did not use
any assumption about the arrival order of the elements in the second phase. Note also that it is
necessary to  know the cardinality $n$ of the matroid in the \msp{} problem (in contrast to
\sbmsp{}) so as to be able to form a sample set $S$ that contains each element with probability $p_s$.

Before presenting our algorithm, we need another simple reduction that
allows the algorithm to assume a certain knowledge about the rank $\rho$ of the underlying matroid
$M=(N, \mathcal{I})$ and the weights of its elements. More precisely, we call an algorithm for
{\sbmsp} \emph{aided} if it assumes access to two additional values $\tilde{\rho}$ and $W$ such
that the considered matroid and these values satisfy:
\begin{compactenum}[(i)]
\item $\tilde{\rho} \geq \rho$, where $\rho$ is the rank of $M$,
\item for every element $e \in N$, $w(e) \in (W / (8\tilde{\rho}), W]$.
\end{compactenum}

\begin{reduction} \label{re:to_aided_algorithms}
  Any $\alpha(\tilde{\rho})$-competitive aided
  algorithm for {\sbmsp}, where $\alpha(\cdot)$ is a non-decreasing function,
  can be transformed efficiently into
a $160\cdot \alpha(4\rho)$-competitive
  (non-aided) algorithm for {\sbmsp}.
\end{reduction}

The main idea of the above reduction is to sample half of the elements, and based on this sample
estimate $W$ and $\rho$. Using these estimates, the aided algorithm is then applied to the
remaining elements whose weight fall inside the range $(W / (8\tilde{\rho}), W]$. The details of the
proof are quite standard and are also deferred to Appendix~\ref{sec:redProofs}.
In the rest of this paper, we focus on obtaining an
$O(\log\log\tilde{\rho})$-competitive aided algorithm
for {\sbmsp}. 

To simplify notation, we use '$+$' and '$-$' for
addition and subtraction of a single element from a set, \eg
$S+e-f=(S\cup \{e\})\setminus \{f\}$.
We denote by $r$
the \emph{rank} function of the matroid $M$, \ie for any subset $S\subseteq N$:
$\rank(S)=\max\{|I| \mid I\in \cI, I\subseteq S\}$
 is the size of a maximum cardinality independent
set in $S$.
Furthermore, the \emph{span} of a subset $S\subseteq N$ is given by $\spn(S) = \{e\in N \mid \rank(S+e) = \rank(S)\}$, and its total weight is given by $w(S) = \sum_{e \in S} w(e)$.
We refer the reader to~\cite{schrijver_2003_combinatorial}
for further matroidal concepts, such as contractions and
restrictions of matroids.

Throughout the paper we assume that the rank $\rho$ of the matroid
under consideration is at least $1$. Clearly, if $\rho=0$, then any
algorithm returning a feasible solution, which means the empty set
in this case, is $1$-competitive.


\section{Bucketing-based algorithm}\label{sec:algo}

%
%

\paragraph{Weight classes and buckets.}Our bucketing-based algorithm distinguishes items based on their weight.
We define $h = \lceil 3 + \log_2
\tilde{\rho} \rceil$ \emph{weight classes} as follows. For every
$i \in \{1,\dots, h\}$,  let
\[
C_i = \left\{e \in N ~\middle|~ w(e) \in \left(\frac{W} {2^{h - i + 1}} , \frac{W}{ 2^{h - i}} \right]\right\}.
\] 
Notice that every element belongs to exactly one class; class $C_1$ contains the lightest elements,
$C_2$ slightly less light elements and so on. Moreover, it is possible to determine upon arrival
 which class an element belongs to. 

Our bucketing-based algorithm takes as input
a \emph{bucketing} $\vec{B} = (B_1, B_2, \ldots, B_b)$,
which is a partition of the weight
classes such that each bucket $B_i$ is the union of
a consecutive set of weight classes. More formally,
every bucket $B_i$ is associated with
two numbers $f(B_i) \leq \ell(B_i)$,
where $f(B_i)$ and $\ell(B_i)$
are the first and last index of the weight
classes composing $B_i$,
respectively, \ie $B_i = \bigcup_{j= f(B_i)}^{\ell(B_i)} C_j$. As the
bucketing $\vec{B}$ partitions the weight classes, its buckets satisfy
\[
1 = f(B_1), \qquad
\ell(B_i) + 1 = f(B_{i + 1})
\;\;\forall 1\leq i < b,  \qquad
\text{and} \qquad \ell(B_b) = h.
\]

For ease of notation, we define $B_{b+1} = \varnothing$ and $f(B_i) = \ell(B_i) = 0$ for every $i
\leq 0$. Furthermore, let $B_{\geq i} = \bigcup_{j = i}^{b} B_j$.

\paragraph{Algorithmic overview.} 
Like many other secretary algorithms,
our bucketing-based
algorithm first observes a random set $S$ containing
each element with probability $1/2$, without selecting any
element of $S$. 
%
Based on the set $S$ we define a matroid $M_i$ for each $i\in \{1,\dots, h\}$ as follows.  $M_i$ is
the matroid obtained from $M$ by first contracting $S\cap B_{\geq i+1}$ and then restricting the
resulting matroid to the elements in $B_i\cap \spn(S\cap B_{\geq i-1})$ (for $i = 1$, we restrict to $B_1$ instead). The restrictions and contractions effectively partition the problem into disjoint matroids, from which we independently pick elements greedily. The aim of this partition is to protect heavy
elements from being spanned by lighter elements. Notice that we restrict to $B_i\cap \spn(S\cap B_{\geq i-1})$ instead of, the perhaps more natural,
$B_i \cap \spn(S \cap B_ {\geq i})$. This choice of the restriction is required for ensuring the
performance guarantee of the algorithm as analyzed in Lemma~\ref{le:meta_algorithm_performance}.
In typical matroid notation, where contractions are
denoted by a slash (`$/$') and restrictions by a vertical bar (`$|$'),
the $M_i$'s are defined as
\[
M_1 = (M / (S\cap B_{\geq 2})) |_{B_1} \qquad \mbox{and} \qquad
M_i = (M / (S\cap B_{\geq i+1}))|_{B_i\cap
  \spn(S\cap B_{\geq i-1})} \qquad \forall i\in \{2,\dots, h\}.
\]
In particular, the ground set of $M_1$ is $N_1 = B_1$,
and the ground set of $M_i$ for $i\in \{2,\dots, h\}$
is $N_i = B_i\cap \spn(S\cap B_{\geq i-1})$.
Furthermore, for $i\in \{1,\dots, h\}$, let $\mathcal{I}_i$
be the collection of $M_i$'s independent sets; hence, $M_i=(N_i,\mathcal{I}_i)$.
We partition the matroids $\{M_i\}_{i = 1}^h$ into two groups according
to the parity of their index.
Let $H_{\odd} = \{i\in \{1,\dots, h\}\mid i \text{ odd}\}$,
and $H_{\even}=\{i\in \{1,\dots, h\}\mid i \text{ even}\}$.

After having observed $S$, our bucketing-based algorithm
chooses at random $H\in \{H_{\odd}, H_{\even}\}$,
and then greedily accepts elements from each matroid
$M_i$ with $i\in H$
as long as independence is preserved within $M_i$.
At the end of the algorithm, a set $T_i\in \mathcal{I}_i$ has been selected for each $i\in H$ and the algorithm
returns $T= \bigcup_{i\in H} {T_i}$, which, as we show later, satisfies $T\in \mathcal{I}$. The
reason why the algorithm restricts itself to either the odd index or even index matroids is
to ensure the feasibility of $T$. The random choice of $H$ averages over the two possibilities, and allows elements of both even and odd buckets a chance to be selected.

\paragraph{Analysis of feasibility.} Algorithm~\ref{alg:BucketsSecretary} is a pseudocode representation
of our bucketing-based algorithm.
\begin{algorithm}[h!t]
\caption{\textsf{Bucketing-based
algorithm}$(\vec{B})$} \label{alg:BucketsSecretary}
Let $S$ be a set containing every element with probability $1/2$.\\
Let $H=H_{\odd}$ with probability $1/2$, and $H=H_{\even}$ otherwise.\\
Let $T_i \leftarrow \varnothing$ for $i\in H$.\\
\For{every element $e \in N \setminus S$ revealed}
{
	Let $i$ be the index of the bucket containing $e$ (\ie $e \in B_i$).\\
	\If{$i \in H$ \textbf{and} $e\in N_i$
     \textbf{and} $e+T_i\in \cI_i$\label{algline:acceptCond}}%
	{
		Add $e$ to $T_i$.\\
	}
}
Return $T= \bigcup_{i\in H} T_i$.
\end{algorithm}
Our first step is to verify that the conditions of accepting
an element, described in line~\ref{algline:acceptCond},
can indeed be verified with the information available
to the algorithm at that point.
To this end we show that the conditions $e\in N_i$
and $e+T_i\in \mathcal{I}_i$ (where $e\in B_i$)
are equivalent to\\[0.2em]
\hspace*{2em}\begin{minipage}{\linewidth}
\begin{compactenum}[(a)]
\item\label{item:cond1} $i=1$ \textbf{or}
  $e\in \spn(S\cap B_{\geq i-1})$ if $i>1$, and
\item\label{item:cond2}
  $e\not\in \spn(T_i\cup (S\cap B_{\geq {i+1}}))$.
\end{compactenum}
\end{minipage}
\smallskip

\noindent The equivalence between $e\in N_i$
and \eqref{item:cond1} for an element $e \in B_i$ follows immediately from
the definition of $N_i$, since $N_1=B_1$
and $N_i = B_i\cap \spn(S\cap B_{\geq i-1})$
for $i\geq 2$.
%
%
Furthermore, Lemma~\ref{lem:indCheck} below shows the equivalence between $e+T_i\in \mathcal{I}_i$
and~\eqref{item:cond2}. Notice
that~\eqref{item:cond1} and~\eqref{item:cond2} only
depend on the element $e$, the set $S$,
and the set $T_i$ of elements selected so far within
$N_i$. Thus, these conditions
can be checked by the algorithm
on line~\eqref{algline:acceptCond} without knowing
the matroid $M$ in advance.

\begin{lemma}\label{lem:indCheck}
Let $S\subseteq N$, $T_i\in \mathcal{I}_i$ and $e\in N_i$
for some $i\in \{1,\dots, h\}$. Then
$T_i + e \in \mathcal{I}_i$ if and only
if $e\not\in \spn(T_i\cup (S\cap B_{\geq i+1}))$.
\end{lemma}
Before proving the lemma,
we recall that since $M_i$ is a contraction and
restriction of $M$, we can use standard
results in matroid theory to express the
rank function $r_i$ of $M_i$ in 
terms of the rank function $r$ of $M$
as follows.
\begin{equation}\label{eq:ri}
r_i(U) = r(U \cup (S\cap B_{\geq i+1})) - r(S\cap B_{\geq i+1})
\qquad \forall U\subseteq N_i.
\end{equation}
\begin{proof}[Proof of Lemma~\ref{lem:indCheck}]
We have $T_i+e\in \mathcal{I}_i$ if and only if
$r_i(T_i+e) = r_i(T_i) +1$, which, by~\eqref{eq:ri},
is equivalent to
\begin{equation*}
r((T_i+e)\cup (S\cap B_{\geq i+1}))
  = 1+r(T_i \cup (S\cap B_{\geq i+1})),
\end{equation*}
which in turn is equivalent to
$e\not\in \spn(T_i\cup (S\cap B_{\geq i+1}))$.
\end{proof}

It is clear that the sets $\{T_i\}_{i \in H}$ constructed by
Algorithm~\ref{alg:BucketsSecretary} indeed satisfy
$T_i\in \mathcal{I}_i$, since the property
$T_i\in \mathcal{I}_i$ is preserved throughout the
algorithm.
Lemma~\ref{lem:composition} below implies that
the returned set $T$ is independent in $M$.

\begin{lemma}\label{lem:composition}
Let $H\in \{H_{\odd},H_{\even}\}$ and let
$I_i\in \mathcal{I}_i$ for $i\in H$. Then
$\bigcup_{i\in H} I_i \in \mathcal{I}$.
\end{lemma}

Whereas a formal
proof of Lemma~\ref{lem:composition} can be
found in Appendix~\ref{sec:missingProofs},
we still want to 
give some intuition and link the lemma to previous work.
For simplicity we focus on the case $H=H_{\even}$,
and assume $h=2k$ is even.
For $i\in \{1,\dots, h\}$, let $A_i=\spn(S\cap B_{\geq i})$.
Hence, $A_1 \supseteq A_2 \supseteq \dots \supseteq A_h$.
Notice that for $i\in H_{\even}$, we have
$M_i = (M/A_{i+1})|_{B_i\cap A_{i-1}}$, and hence,
$M_i$ is a restriction of the matroid $M_i'=(M/A_{i+1})|_{A_{i-1}}$.
Therefore, any independent set $I_i$ of $M_i$ is also
independent in $M_i'$.
However, for any sequence of matroids defined by
$M_i'=(M/A_{i+1})|_{A_{i-1}}$
where the sets $A_i$ form a chain
$A_1\supseteq A_3 \supseteq \dots \supseteq A_{h-1}$,
one can easily verify that if $I_i$ is independent in $M_i$
for all $i\in H_{\even}$,
then $\bigcup_{i\in H_{\even}} I_i$ is independent in $M$.
Actually, this reasoning---for a chain formed by different
sets $A_i$---has
already been used in the context of the matroid secretary
problem by Soto~\cite{soto_2013_matroid}.

\begin{corollary}
Algorithm~\ref{alg:BucketsSecretary} returns
an independent set $T\in \mathcal{I}$.
\end{corollary}
\begin{proof}
This is an immediate result of Lemma~\ref{lem:composition}
and the fact that each set in $\{T_i\}_{i \in H}$ constructed by Algorithm~\ref{alg:BucketsSecretary}
is independent in its corresponding matroid $M_i$.
\end{proof}

\paragraph{Analysis of performance guarantees.} Whereas decomposition approaches similar to the above
have already been
used~(see, e.g.,~\cite{jaillet_2013_advances,soto_2013_matroid}),
a main novelty of our approach is the way we lower bound
the likelihood of elements to be selected.
It turns out that selection
probabilities can be elegantly lower bounded in terms of the
following probabilities
\begin{equation*}
p_{e,i} = \Pr[e\in \spn(S\cap C_{\geq i})\mid e\not\in S]
\qquad \forall e\in N, i\in \{1,\dots, h\},
\end{equation*}
where $C_{\geq i} = \bigcup_{j=i}^h C_j$.
We remind the reader that $C_j$ denotes the $j$-th
weight class.
For consistency,
we also define $p_{e,0}=1$ for every $e\in N$.
Notice that $p_{u,i}$ is non-increasing in $i$.
The following two lemmata describe lower bounds
on the selection probabilities.
We recall that for any bucket $B_i$, the expression
$f(B_i)\in\{1,\dots, h\}$ denotes the lowest
index of all weight classes contained in $B_i$. 

\begin{lemma} \label{le:meta_algorithm_performance}
For every element $e \in B_i$,
\[
	\Pr[e \in T]
	\geq
	\frac{p_{e, f(B_{i-1})} - p_{e, f(B_i)}}{4}
	.
\]
\end{lemma}
\begin{proof}
Let $G$ be the event that $i \in H$. If $G$ does not occur, then,
clearly, $e \not \in T$. Thus, in the rest of the proof we
assume $G$ occurs, and implicitly condition all the expectations
on this assumption. As a result, the lower bound we prove applies
in fact to $\Pr[e \in T \mid G] = 2 \cdot \Pr[e \in T]$.

Recall that $e$ gets selected by Algorithm~\ref{alg:BucketsSecretary}
if $e\not\in S$ and it obeys two conditions: $e\in N_i$
and $e + T_i^e \in \cI_i$, where 
$T_i^e$ is the set $T_i$ immediately before $e$ is revealed.

If $i = 1$, then the first condition always holds
since then $e\in B_1 = N_1$, \ie it holds
with a probability of $1 = p_{e, f(B_{i - 1})} = p_{e,0}$.
Otherwise, $N_i=B_i\cap \spn(S\cap B_{\geq i-1})$,
which implies, together with $e \in B_i$, the equality:
\begin{align*}
	\Pr[e\in N_i\mid e\not\in S]
  =
  \Pr[e \in \spn(S \cap B_{\geq i -1}) \mid e \not \in S]
	= p_{e, f(B_{i-1})}.
\end{align*}
Hence, in both cases, the first condition is satisfied (conditioned on $e \not \in S$) with
probability $p_{e,f(B_{i-1})}$.

We now proceed to analyze the probability of the second condition $e + T^e_i \in \cI_i$, which, by
Lemma~\ref{lem:indCheck}, can be equivalently stated as 
$e \not \in \spn(T_i^e \cup (S \cap B_{\geq i + 1}))$.
Let $\bar{S}=N\setminus S$. Then,
\begin{align*}
 \Pr[e + T^e_i \not \in \cI_i \mid e \not \in S] & =  	\Pr[e \in \spn(T_i^e \cup (S \cap B_{\geq i + 1})))
    \mid e \not \in S] \\
      	&\leq
	\Pr[e \in \spn((\bar{S} \cap B_i) \cup (S \cap B_{\geq i + 1}))
    \mid e \not \in S]\\
        &= 
        \Pr[e \in \spn(S\cap B_{\geq i})\mid e\not\in S] 
        = p_{e,f(B_i)},
\end{align*}
where the inequality follows from $T_i^e\subseteq \bar{S}\cap B_i$,
and the second equality follows from the fact that
$\bar{S}\cap B_i$ and $S\cap B_i$ are identically distributed and both sets are independent of $S\cap B_{\geq {i+1}}$.

In conclusion, $e$ is accepted (conditioned on $G$) with probability 
\begin{align*}
	\Pr[e \not \in S \text{ \textbf{and} }
    e \in N_i \textbf{ and } e + T_i^e \in \cI_i] 
%
	&\geq 
	\Pr[e \not \in S] \cdot \Big(\Pr[e\in N_i \mid e \not \in S] 
      - \Pr[e + T^e_i \not \in \cI_i 
           \mid e \not \in S]\Big)\\
%
%
%
	&\geq
	\frac{p_{e, f(B_{i-1})} - p_{e, f(B_i)}}{2}
	. \qedhere
\end{align*}
\end{proof}

\begin{lemma}\label{lem:metaSelectMany}
For every $i\in \{1,\dots, b\}$,
\begin{equation*}
\E[|T\cap B_i|] \geq \frac{1}{4}
\sum_{e\in B_i\cap \OPT} p_{e,f(B_{i-1})}
\geq
\frac{1}{4}
\sum_{e\in B_i\cap \OPT} p_{e,f(B_i)}
.
\end{equation*}
\end{lemma}
The proof of the above Lemma is deferred to
Appendix~\ref{sec:missingProofs}.
Notice that the second inequality of
Lemma~\ref{lem:metaSelectMany} follows immediately
from the fact that $p_{e,i}$ is non-increasing in $i$.



\section{Full algorithm} \label{sec:full_algo}

Our full algorithm chooses a random bucketing~$\vec{B}$
according to a well-chosen distribution and then calls
the bucketing-based algorithm with~$\vec{B}$ as input.
The random bucketing~$\vec{B}$ is chosen such that all
buckets $B_i$ have the same length $\ell(B_i)-f(B_i)$,
i.e., each contains the same number of weight classes,
except for, possibly, the first and last buckets, which may
be shorter. This common length of the buckets is
chosen to be a power of two, $2^\tau$,
drawn uniformly at random from all lengths
that are powers of two and lie between $1$ and
the first power of two that is at least $h+1$.
In other words, $\tau$ is drawn uniformly at random
from $0,1,\dots, \lceil \log_2(h+1) \rceil$.
Furthermore, a uniform random shift
$\Delta\in \{0,\dots, 2^\tau -1\}$ defines
where the first bucket ends: the first
bucket contains the lightest $2^\tau -\Delta$
weight classes, and, as described, each following
bucket bundles $2^\tau$ weight classes until the
last bucket which may have a shorter length.
Figure~\ref{fig:randBucketing}
exemplifies the bucketing.

\begin{figure}[h!]
\resizebox{\linewidth}{!}{%
\begin{tikzpicture}[font={\large}]

\begin{scope}[every node/.style={circle, draw,inner sep=0pt,minimum size=8mm}]

\foreach \i in {1,...,6} {
   \node at (\i,0) {$C_\i$};
}
\foreach \i in {7,...,10,13,14,15} {
  \node at (\i,0) {};
}
\node at (16,0) {$C_h$};

\begin{scope}[every node/.append style={dotted}]
\foreach \i in {-2,-1,0,17,18} {
 \node at (\i,0) {};
}
\end{scope}

\draw[loosely dotted,very thick] (10.6,0) -- (12.4,0);

\end{scope}

\def\eps{0.05}
\begin{scope}[xshift=-0.5cm, yshift=-0.5cm,rounded corners]
\foreach \i in {-2,2,6,15}{
\draw (\i+\eps,0) rectangle (\i+4-\eps,1);
}

\begin{scope}
\clip (10,-0.1) rectangle (11,1.1);
\draw (10+\eps,0) rectangle (10+4-\eps,1);
\end{scope}

\begin{scope}
\clip (13,-0.1) rectangle (15,1.1);
\draw (11+\eps,0) rectangle (11+4-\eps,1);
\end{scope}

\end{scope} 

\begin{scope}
\node[above] at (-0.5,0.5) {$B_1=\{C_1\}$};
\node[above] at (3.5,0.5) {$B_2=\{C_2,\dots, C_5\}$};
\node[above] at (7.5,0.5) {$B_3=\{C_6,\dots, C_9\}$};

\node[above] at (16.5,0.5) {$B_{\lceil (h+\Delta)/2^\tau\rceil}$};
\end{scope}

\begin{scope}[thick]
\draw (0.5,0.6) -- (0.5,-1);
\draw (16.5,0.6) -- (16.5,-1);
\end{scope}

\begin{scope}[thick,yshift=-0.7cm,stealth-stealth,
shorten >= 2pt,shorten <= 2pt]
\draw (1.5,0) -- node[below] {$2^{\tau}=4$} (5.5,0);
\draw (-2.5,0) -- node[below] {$\Delta=3$} (0.5,0);
\end{scope}

\end{tikzpicture}%
}
\caption{An illustration of how
the bucketing is done for $\Delta=3$
and $\tau=2$.}
\label{fig:randBucketing}
\end{figure}
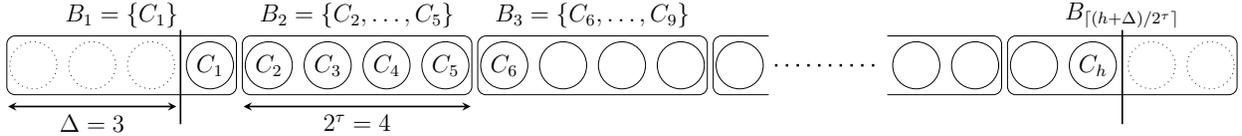

Algorithm~\eqref{alg:LogLogSecretary} is a pseudocode
description of our algorithm.
For ease of notation, we assume in the algorithm that $C_i = \varnothing$ whenever $i > h$ or $i \leq 0$. This assumption allows us to write expressions of the form $\bigcup_{j = i}^{i'} C_j$ where $i' > h$ and/or $i \leq 0$.
Also, recall that, by definition, $f(B_i) = \ell(B_i) = 0$ for every $i \leq 0$.

\begin{algorithm}[h!t]
\caption{Matroid Secretary Algorithm} \label{alg:LogLogSecretary}
Let $\tau$ be a uniformly random number from the range $0,1, \ldots, \lceil \log_2 (h + 1) \rceil$.\\
Let $\Delta$ be a uniformly random number from the range $0, 1, \ldots, 2^\tau - 1$.\\
Let $\vec{B}$ be a bucketing with $\lceil (h + \Delta) / 2^\tau \rceil$ buckets, where bucket $B_i$ is defined as follows:
\[
	B_i = {\textstyle \bigcup_{j = 2^\tau \cdot (i - 1) - \Delta + 1}^{2^\tau \cdot i - \Delta}} C_j.
\]\\
Execute Algorithm~\ref{alg:BucketsSecretary}
with bucketing $\vec{B}$
and return resulting $T\in \mathcal{I}$.
\end{algorithm}

To analyze Algorithm~\ref{alg:LogLogSecretary} we
leverage the two lower bounds on selection
probabilities derived for
the bucketing-based algorithm, \ie
Lemma~\ref{lem:metaSelectMany} and
Lemma~\ref{le:meta_algorithm_performance}.
More precisely, we use Lemma~\ref{lem:metaSelectMany}
to analyze the case when Algorithm~\ref{alg:LogLogSecretary}
runs with $\tau=0$, and employ
Lemma~\ref{le:meta_algorithm_performance}
for $\tau\geq 1$.
We start with the case $\tau=0$.
As usual, let $T$ be the set returned
by Algorithm~\ref{alg:LogLogSecretary}.

\begin{lemma} \label{lem:accept_propT0}
For every $i\in \{1,\dots, h\}$,
\begin{equation*}
\E\Big[|T\cap C_i| ~\Big|~ \tau = 0\Big]
  \geq \frac{1}{4}\sum_{e\in C_i\cap \OPT} p_{e,i}.
\end{equation*}
\end{lemma}
\begin{proof}
For $\tau=0$, the bucketing $\vec{B}$ consists
of the weight classes, i.e., $\vec{B}=\{C_1,\dots, C_h\}$.
The result then immediately follows
from Lemma~\ref{lem:metaSelectMany}.
\end{proof}

\begin{lemma} \label{lem:accept_propTgeq1}
For every $i\in \{1,\dots, h\}$
and $e\in C_i$,
\begin{equation*}
\Pr[e\in T \mid \tau \geq 1 ]
\geq 
\frac{1 - p_{e, i}}{8 \cdot \lceil \log_2 (h + 1) \rceil}.
\end{equation*}
\end{lemma}
Before proving Lemma~\ref{lem:accept_propTgeq1}, we show how to
derive from
Lemma~\ref{lem:accept_propT0} and Lemma~\ref{lem:accept_propTgeq1}
an upper bound of $O(\log \log \tilde{\rho})$ on the competitive ratio of Algorithm~\ref{alg:LogLogSecretary}.
Combining both lemmata we first obtain the following.
\begin{corollary}\label{cor:selectInWC}
For every $i\in \{1,\dots, h\}$,
\begin{equation*}
\E\big[|T\cap C_i|\big] \geq
  \frac{|C_i\cap \OPT|}{8(\lceil \log_2(h+1)\rceil+1)} .
\end{equation*}
\end{corollary}
\begin{proof}
By summing the inequality of Lemma~\ref{lem:accept_propTgeq1}
over all elements $e\in C_i$ we get
\begin{equation}\label{eq:sumT0}
\E\Big[|T\cap C_i| ~\Big|~ \tau\geq 1\Big]
  \geq \frac{\sum_{e\in C_i} (1-p_{e,i})}{8 \lceil \log_2(h+1)\rceil}
  \geq \frac{\sum_{e\in C_i \cap \OPT} (1-p_{e,i})}{8 \lceil \log_2(h+1)\rceil}
  .
\end{equation}
It remains to combine~\eqref{eq:sumT0} with
Lemma~\ref{lem:accept_propT0} to obtain
\begin{align*}
\E\big[|T\cap C_i|\big] &=
  \Pr[\tau = 0] \cdot \E\Big[ |T\cap C_i| ~\Big|~ \tau=0 \Big]
  + \Pr[\tau \geq 1] \cdot \E\Big[ |T\cap C_i| ~\Big|~ \tau \geq 1\Big]\\
&= \frac{1}{\lceil \log_2(h+1)\rceil+1}
\cdot
\left(
  \E\Big[ |T\cap C_i| \Big| \tau=0 \Big]
  + \lceil\log_2(h+1) \rceil \cdot
    \E\Big[ |T\cap C_i| \Big| \tau \geq 1 \Big]
\right)\\
&\geq \frac{1}{\lceil \log_2(h+1)\rceil+1}
\cdot
\left(
\frac{\sum_{e\in C_i\cap \OPT} p_{e,i}}{4}
+ \frac{\sum_{e\in C_i\cap \OPT} (1-p_{e,i})}{8}
\right)\\
&\geq \frac{|C_i\cap \OPT|}{8 (\lceil \log_2(h+1)\rceil+1)}.
\qedhere
\end{align*}
\end{proof}

A lower bound on the competitiveness of
Algorithm~\ref{alg:LogLogSecretary} can now easily be
derived from Corollary~\ref{cor:selectInWC}.
The following Theorem states this lower bound, and shows that it
implies Theorem~\ref{thm:mainThm}.
\begin{theorem}
Algorithm~\ref{alg:LogLogSecretary} is an aided algorithm
for {\sbmsp} whose competitive ratio is at most
$16 \cdot (\lceil \log_2 (h + 1) \rceil+1)
 = 16 \cdot (\lceil \log_2 (4 + \lceil \log_2 \tilde{\rho} \rceil) \rceil+1)
  \leq 16 [\log_2\log_2 (\max\{\tilde{\rho}, 4\}) + 5]$.
Hence, by Reductions~\ref{re:timed} and \ref{re:to_aided_algorithms} there exists
a (non-aided) $2560[\log_2 \log_2 (4\rho) + 5]$-competitive
algorithm for {\msp}.
\end{theorem}
\begin{proof}
By Corollary~\ref{cor:selectInWC} and the fact
that the weights of the elements within each weight class differ
by a factor of at most $2$, the expected weight of the elements selected from
any $C_i$ is at least a
$[16(\lceil \log_2(h+1)\rceil +1)]^{-1}$-fraction
of $w(C_i\cap \OPT)$. By summing this bound
over all weight classes, we obtain that
Algorithm~\ref{alg:LogLogSecretary} is
$16(\lceil \log_2(h+1)\rceil +1)$-competitive.
\end{proof}

Hence, it remains to prove Lemma~\ref{lem:accept_propTgeq1}.

\begin{proof}[Proof of Lemma~\ref{lem:accept_propTgeq1}]
One can think of the expression $1-p_{e,i}=p_{e,0}-p_{e,i}$
in the statement of the
lemma as the change in $p_{e,j}$ as $j$ goes from $j=i$ to $j=0$. 
We analyze this change by considering smaller intervals.
Let $k=\lceil \log_2(i+1)\rceil$, and define
$a_j = i - 2^j + 1$ for $j\in \{0,\dots, k\}$.
Notice that $a_j \in \{1, \ldots, i\}$ for $j < k$ and $a_k \leq 0$.
Figure~\ref{fig:bucketingProof} illustrates our
choice of the indices $a_i$. Observe that the distance $a_{j - 1} - a_j$ doubles each time $j$ increases by $1$.
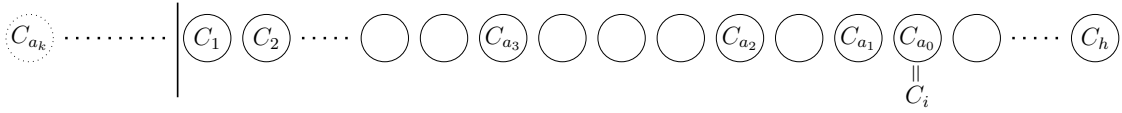
\begin{figure}[h!t]
\hspace{1cm}\resizebox{\linewidth}{!}{%
\begin{tikzpicture}[font={\large}]

\begin{scope}[every node/.style={circle, draw,inner sep=0pt,minimum size=8mm}]

\foreach \i in {1,2} {
   \node at (\i,0) {$C_\i$};
}
\foreach \i in {4,...,14} {
  \node at (\i,0) {};
}
\node at (16,0) {$C_h$};

\begin{scope}[every node/.append style={dotted}]
\foreach \i in {-2} {
 \node at (\i,0) {};
}
\end{scope}

\node[white] at (18,0) {};

\end{scope}

\node at (13,0) {$C_{a_0}$};
\node[yshift=-6mm,rotate=90] (eq) at (13,0) {$=$};
\node[below=2pt] at (eq) {$C_i$};

\node at (12,0) {$C_{a_1}$};
\node at (10,0) {$C_{a_2}$};
\node at (6,0) {$C_{a_3}$};
\node at (-2,0) {$C_{a_k}$};

\draw[loosely dotted,very thick] (-1.4,0) -- (0.4,0);
\draw[loosely dotted,very thick] (2.6,0) -- (3.4,0);
\draw[loosely dotted,very thick] (14.6,0) -- (15.4,0);

\begin{scope}[thick]
\draw (0.5,0.6) -- (0.5,-1);
\end{scope}

\def\eps{0.05}

\end{tikzpicture}%
}
\caption{An illustration of how
the indices $a_i$ are chosen.}
\label{fig:bucketingProof}
\end{figure}

By defining, for ease of notation,
$p_{e,a_k}=1$, we can write
$1-p_{e,i} = \sum_{j=1}^k (p_{e,a_j} - p_{e,a_{j-1}})$.
%
We show the following:
\begin{equation}\label{eq:intAnalysis}
\Pr[e\in T \mid \tau=j] \geq \frac{1}{8} (p_{e,a_j}-p_{e,a_{j-1}})
\qquad \forall j\in \{1,\dots, k\}.
\end{equation}
First, observe that the lemma follows easily
from the above inequality.
\begin{align*}
\Pr[e\in T\mid \tau \geq 1] &\geq \sum_{j=1}^k
 \underbrace{\Pr[\tau=j \mid \tau\geq 1]}_{=\frac{1}%
        {\lceil\log_2(h+1)\rceil}}
  \cdot \Pr[e\in T \mid \tau=j]\\
&\geq \frac{1}{8\lceil\log_2(h+1)\rceil}\cdot
  \sum_{j=1}^k (p_{e,a_j} - p_{e,a_{j-1}})
 = \frac{1-p_{e,i}}{8\lceil\log_2(h+1)\rceil}.
\end{align*}
Hence, it remains to show~\eqref{eq:intAnalysis}.
Fix some $j\in \{1,\dots, k\}$, and consider an execution
of Algorithm~\ref{alg:LogLogSecretary} with $\tau =j$.
Let $B_s$ be the (random) bucket containing
$C_i$. 
We denote by $G$ the event that the random shift
$\Delta$ of Algorithm~\ref{alg:LogLogSecretary} is such
that $f(B_s)\in \{a_{j-1},a_{j-1}+1,\dots, a_0\}$
(see Figure~\ref{fig:bucketingProof2} for an illustration).
%
%
%
\begin{figure}[ht]
\hspace{1.5cm}\resizebox{0.9\linewidth}{!}{%
\begin{tikzpicture}[font={\large}]

\begin{scope}[every node/.style={circle, draw,inner sep=0pt,minimum size=8mm}]

\foreach \i in {2,...,15} {
  \node at (\i,0) {};
}

\begin{scope}[every node/.append style={dotted}]
\end{scope}

\node[white] at (18,0) {};

\end{scope}

\node at (13,0) {$C_{a_0}$};
\node[yshift=-6mm,rotate=90] (eq) at (13,0) {$=$};
\node[below=2pt] at (eq) {$C_i$};

\node at (12,0) {$C_{a_1}$};
\node at (10,0) {$C_{a_2}$};
\node[yshift=-6mm,rotate=90] (eq) at (10,0) {$=$};
\node[below=2pt] at (eq) {$C_{a_{j-1}}$};

\node at (6,0) {$C_{a_3}$};
\node[yshift=-6mm,rotate=90] (eq) at (6,0) {$=$};
\node[below=2pt] at (eq) {$C_{a_{j}}$};


\draw[loosely dotted,very thick] (0.6,0) -- (1.4,0);
\draw[loosely dotted,very thick] (15.6,0) -- (16.4,0);

\begin{scope}[thick]
\end{scope}

\def\eps{0.05}
\begin{scope}[xshift=-0.5cm, yshift=-0.5cm,rounded corners]
\foreach \i in {4}{
\draw (\i+\eps,0) rectangle (\i+8-\eps,1);
}

\begin{scope}
\clip (1.5,-0.1) rectangle (4,1.1);
\draw (-4+\eps,0) rectangle (-4+8-\eps,1);
\end{scope}

\begin{scope}
\clip (12,-0.1) rectangle (16.5,1.1);
\draw (12+\eps,0) rectangle (12+8-\eps,1);
\end{scope}

\end{scope} 

\begin{scope}
\node[above] at (7.5,0.5) {$B_{s-1}$};
\node[above] at (15.5,0.5) {$B_s$};
\end{scope}

\end{tikzpicture}%
}
\caption{For $j=3$, this is a realization
of the random bucketing where the event $G$ occurred.
$G$ occurs when $a_2\leq f(B_s)$, \ie in $4$ offsets out of the $8$ possible offsets in this case.}
\label{fig:bucketingProof2}
\end{figure}
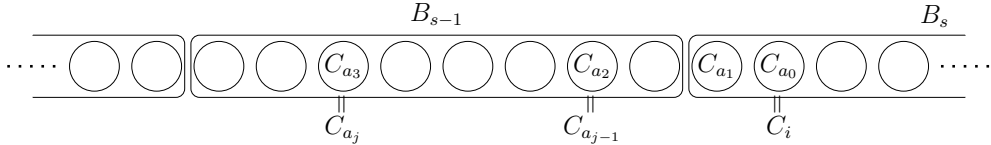
Since the shift $\Delta$ is chosen
uniformly at random
within $0\leq \Delta \leq 2^\tau-1$, we get
\begin{equation*}
\Pr[G] = \frac{a_0 - a_{j-1} +1}{2^\tau} = \frac{2^{j-1}}{2^\tau}
= \frac{1}{2},
\end{equation*}
where the last equality follows by $\tau=j$.
Notice that when $G$ occurs, we have
$f(B_s)\geq a_{j-1}$ and
$f(B_{s-1}) = f(B_s) - 2^\tau \leq i-2^\tau < a_j$.
Hence, we obtain by Lemma~\ref{le:meta_algorithm_performance},
\begin{align*}
\Pr[e\in T\mid \tau =j, G] \geq
\frac{1}{4}(p_{e,f(B_{s-1})} - p_{e,f(B_s)})
\geq \frac{1}{4}(p_{e,a_j} - p_{e,a_{j-1}}),
\end{align*}
where the second inequality follows from
the fact that $p_{e,t}$ is non-increasing in $t$.
Finally,~\eqref{eq:intAnalysis} now easily follows:
\begin{equation*}
\Pr[e\in T\mid \tau =j] \geq
  \Pr[G]\cdot \Pr[e\in T \mid \tau=j, G]
  \geq \frac{1}{8}(p_{e,a_j}-p_{e,a_{j-1}}).
	\qedhere
\end{equation*}

\end{proof}

\bibliographystyle{plain}
\bibliography{lit}

\appendix

\section{Formal Proofs of Reductions}\label{sec:redProofs}

In this section we give the formal proofs of the reductions from Section~\ref{sec:preliminaries}. We
start by Reduction~\ref{re:timed} and then continue with Reduction~\ref{re:to_aided_algorithms}.

\begin{repreduction}{re:timed}
Any $\alpha$-competitive algorithm for {\sbmsp} can be transformed efficiently into an order-oblivious $\alpha$-competitive algorithm for {\msp}.
\end{repreduction}

Before proving the reduction itself, let us prove a technical helper lemma (similar lemmata can be found, \eg in~\cite{kleinberg_2005_multiple-choice}).

\begin{lemma} \label{le:random_sample_from_permutation}
Given a random permutation $\pi$ of $N$ and a binomial random variable $X \sim B(n, p)$, the set $S$ of the first $X$ elements in $\pi$ contains every element $e \in N$ with probability $p$, independently.
\end{lemma}
\begin{proof}
Consider a set $S \subseteq N$ containing every element of $e \in N$ with probability $p$, independently. Let us construct a permutation $\pi$ of $N$ from $S$ as follows. The first $|S|$ elements of $\pi$ are a uniformly random permutation of $S$, and the $n - |S|$ other elements of $\pi$ are a uniformly random permutation of $N \setminus S$. Also, let $X = |S|$. Notice that $X$ is distributed according to $B(n, p)$, and by definition $S$ contains exactly the first $X$ elements of $\pi$.

By symmetry, $\pi$ is a uniformly random permutation even when conditioned on $X$. Hence, $X$ and $\pi$ are independent, which completes the proof of the lemma.
\end{proof}

We are now ready to prove Reduction~\ref{re:timed}.

\begin{algorithm}[ht]
\caption{\textsf{Algorithm for Reduction~\ref{re:timed}}($ALG$)} \label{alg:AllowSbmsp}
Let $p_s$ be the sampling probability reported by $ALG$.\\
Sample $X \sim B(n, p_s)$.\\
Let $S$ be the set of the first $X$ elements in the random input permutation.\\
Run $ALG$ with $S$ as the first stage and $N \setminus S$ as the second stage.\\ 
\end{algorithm}

\begin{proof}
  Fix an $\alpha$-competitive algorithm $ALG$ for {\sbmsp}. Algorithm~\ref{alg:AllowSbmsp} is an
  algorithm for {\msp} which uses $ALG$. The algorithm begins by sampling $X \sim B(n, p_s)$ and then collects
  the first $X$ elements of the random input permutation into a set $S$. By
  Lemma~\ref{le:random_sample_from_permutation}, $S$ contains every element $e \in N$ with
  probability $p_s$, independently. Hence, $S$ can be used as the input for the first stage of
  $ALG$. The rest of the elements (\ie the elements of $N \setminus S$) are then passed, when
  revealed, to $ALG$ as the second phase input. Since Algorithm~\ref{alg:AllowSbmsp} produces a
  solution as valuable as $ALG$, it is also $\alpha$-competitive. It is important to stress that
  Algorithm~\ref{alg:AllowSbmsp} is order oblivious since $ALG$ assumes nothing about the order in
  which the elements are revealed in the
  second phase; in particular, they might be  ordered adversarially.
\end{proof}

The rest of this section is devoted to proving the following reduction.

\begin{repreduction}{re:to_aided_algorithms}
Any $\alpha(\tilde{\rho})$-competitive aided
algorithm for {\sbmsp}, where $\alpha(\cdot)$ is a non-decreasing function,
can be transformed efficiently into 
a $160\cdot \alpha(4\rho)$-competitive
  (non-aided) algorithm for {\sbmsp}.
\end{repreduction}

\SetKwIF{With}{OtherwiseWith}{Otherwise}{with}{do:}{otherwise with}{otherwise}{}
\begin{algorithm}[ht]
\caption{\textsf{Algorithm for Reduction~\ref{re:to_aided_algorithms}}($ALG$)} \label{alg:AllowAided}
Let $S$ be a set containing every element $e \in N$ with probability $1/2$.\\
Let $W \leftarrow \max_{e \in S} w(e)$ and $\tilde{\rho} \leftarrow 4 \cdot \rank(S)$.\\
\With{probability $1/2$}
{
	Pick the first element of $N \setminus S$ whose value is at least $W$.\\
}
\Else
{
	Run $ALG$ with $\tilde{k}$ and $W$ on the elements of $N \setminus S$, ignoring elements of weight $W / (8\tilde{\rho})$ or less.
}
\end{algorithm}

Reduction~\ref{re:to_aided_algorithms} is implemented by Algorithm~\ref{alg:AllowAided}. Observe that Algorithm~\ref{alg:AllowAided} can be implemented as an {\sbmsp} algorithm if $ALG$ can (if $ALG$ declares a sampling probability $p_s$, then Algorithm~\ref{alg:AllowAided} declares a sampling probability of $(1 + p_s)/2$). To prove Reduction~\ref{re:to_aided_algorithms}, we show that Algorithm~\ref{alg:AllowAided} is $160 \cdot \alpha(4k)$-competitive whenever $ALG$ is an $\alpha(\tilde{k})$-competitive aided algorithm. The proof follows immediately from the following claims. The first claim analyzes the case where there exists a single element which has, alone, much of the weight of $\OPT$.

\begin{claim} \label{cl:classical_result}
If there exists an element $e \in N$ such that $w(e) \geq w(\OPT) / 20$, then Algorithm~\ref{alg:AllowAided} is at least $160$-competitive.
\end{claim}
\begin{proof}
Let $e_1$ and $e_2$ be the elements with the highest and second highest weights, respectively. Consider the event when $e_1 \not \in S$, $e_2 \in S$ and Algorithm~\ref{alg:AllowAided} decides not to execute $ALG$. Clearly, this event happens with probability $1/8$. When it happens, Algorithm~\ref{alg:AllowAided} is guaranteed to pick $e_1$, and get a value of $w(\OPT) / 20$.
\end{proof}

The two next claims analyze the case where no single element is very valuable. Let $V \subseteq N$
be the set of elements of weight strictly more than $W/(8\tilde{\rho})$, and let $G$ be the event that all
the following happens:
\begin{center}
\begin{tabular}{rlrl}
	(i)& Algorithm~\ref{alg:AllowAided} executes $ALG$. &\qquad
	(iii)& $|S \cap \OPT| \geq \rho / 4$. \\[1mm]
	(ii)& $w((\OPT \cap V) \setminus S) \geq w(\OPT)/8$. &
\qquad	(iv)& The heaviest element of $N$ is in $S$.
\end{tabular}
\end{center}

\begin{claim}
Conditioned on $G$ happening, Algorithm~\ref{alg:AllowAided} outputs a solution of expected value at least $(8 \cdot \alpha(4{\rho}))^{-1} \cdot w(\OPT)$.
\end{claim}
\begin{proof}
Fix an arbitrary set $S$ for which $G$ happens. We show that Algorithm~\ref{alg:AllowAided} outputs a solution of expected value at least $(8 \cdot \alpha(\tilde{\rho}))^{-1} \cdot w(\OPT)$ conditioned on any such set $S$.

Notice that $ALG$ observes the instance of {\sbmsp} corresponding to the matroid $M|_{V \setminus
  S}$, i.e., the matroid $M$ restricted to those elements not appearing in $S$ that have weight at least
$W/(8 \tilde \rho)$. Let us
verify that the values $W$ and $\tilde{\rho}$ supplied to $ALG$ are appropriate for this
instance. Since $W$ is the weight of the heaviest element in $N$ (by (iv)), the weight of every element $e \in
V \setminus S$ is within the range $(W / (8\tilde{\rho}), W]$. On the other hand, using that (iii) holds
\[
	\tilde{\rho}
	=
	4 \cdot \rank(S)
	\geq
	4 \cdot \rank(S \cap \OPT)
	=
	4 \cdot |S \cap \OPT|
	\geq
	\rho
	.
\]

Since $ALG$ is $\alpha(\tilde{\rho})$-competitive when supplied with appropriate $\tilde{\rho}$ and
$W$ values, it is guaranteed to pick, in expectation, a solution of value at least $w((\OPT \cap V)
\setminus S) / \alpha(\tilde{\rho}) \geq (8 \cdot \alpha(\tilde{\rho}))^{-1} \cdot w(\OPT)$, where we
used that (ii) holds for the inequality. The claim now follows since:
\[
	\tilde{\rho}
	=
	4 \cdot \rank(S)
	\leq
	4\rho
	.
	\qedhere
\]
\end{proof}

The above claim shows that whenever $G$ happens, Algorithm~\ref{alg:AllowAided} performs well. We
finish the proof of the reduction by lower bounding the probability of $G$.
\begin{claim}
If $w(e) \leq w(\OPT) / 20$ for every element $e \in N$, then $G$ happens with probability at least $1/20$.
\end{claim}
\begin{proof}
Let $G'$ be the event that (iii) and (iv) hold and in addition $w(\OPT \setminus S) \geq w(\OPT)/7$.
We continue by first lower bounding $\Pr[G\mid G']$ by $1/2$ and then $\Pr[G']$  by $1/10$ which in turn implies the
claim as $\Pr[G] = \Pr[G\mid G'] \Pr[G']$.

To bound $\Pr[G \mid G']$ notice that it equals the probability that Algorithm~\ref{alg:AllowAided}
executes $ALG$ and the probability that $w((\OPT \cap V) \setminus S) \geq w(\OPT)/8$ conditioned on $G'$. Clearly the
probability of $ALG$ being executed is $1/2$ and independent of the event $G'$. We shall now prove
that $G'$ in fact implies $w((\OPT \cap V) \setminus S) \geq w(\OPT)/8$, and therefore, $\Pr[G \mid G']
= 1/2$.
Observe that by (iii):
\[
	\tilde{\rho}
	=
	4 \cdot \rank(S)
	\geq
	4 \cdot \rank(S \cap \OPT)
	=
	4 \cdot |S \cap \OPT|
	\geq
	\rho
	,
\]
and thus,
\begin{align*}
	w((\OPT \cap V) \setminus S)
	\geq{} &
	w(\OPT \setminus S) - w(\OPT \setminus V)
	\geq
	\frac{w(\OPT)}{7} - \rho \cdot \frac{W}{8\tilde{\rho}}\\
	\geq{} &
        \frac{w(\OPT)}{7} - \rho \cdot \frac{w(\OPT)}{160\rho}
	\geq
	\frac{w(\OPT)}{8}
	.
\end{align*}
The second inequality follows from the fact that $w(\OPT \setminus V)$ contains at most $\rho$ elements, each having a weight of at most $W/(8\tilde{\rho})$; the third follows from the inequality $\tilde{\rho} \geq \rho$
and the assumption of the claim, \ie  that any element has weight
at most $w(\OPT)/20$ and therefore $W\leq w(\OPT)/20$.

Having proved $\Pr[G \mid G'] = 1/2$, we continue by  lower bounding $\Pr[G']$. We shall do so
by upper bounding the probability that each of its conditions is violated, and then applying the
union bound.  
\begin{description}
\item[Condition (iv):] It is clear that the   the heaviest element of $N$ is \emph{not} in  $S$ with probability $1/2$ because each element is in $S$
  with probability $1/2$.
\item[Condition (iii):]  
By the conditions of the claim, $\OPT$ must contain at least $20$ elements. Since each element
appears in $S$ with probability $1/2$, we get by the Chernoff bound\footnote{To be precise, we use
  that $\Pr[X \leq (1-\delta) \mu) \leq e^{-\frac{\delta^2 \mu}{2}}, 0 < \delta <1$. In our setting
  $X: =|S\cap \OPT|,\mu:= |\OPT|/2 = \rho/2$ and  $\delta:= 1/2$.} that
\[
	\Pr[\neg (iii)] = \Pr[|S \cap \OPT| < \rho /4]
	\leq
	e^{-|\OPT|/16}
	\leq
	e^{-20/16}
	\leq
	0.287
	.
\]

\item[Condition $\boldsymbol{w(\mathrm{OPT} \setminus S) \geq w(\mathrm{OPT})/7}$:] Let us upper bound the probability
  that this condition does not hold, i.e.,  $\Pr[w(\OPT \setminus S) < w(\OPT)/7]$. Notice that
\[
	\mathbb{E}[w(\OPT \setminus S)]
	=
	\frac{w(\OPT)}{2}
	,
\]
and because $S$ contains each element with probability $1/2$, independently,
\begin{align*}
	\var[w(\OPT \setminus S)]
	={} &
	\sum_{e \in \OPT} \var[w(\{e\} \setminus S)]
	=
	\sum_{e \in \OPT} \frac{[w(e)]^2}{4},
\end{align*}
which in turn, as $w(e) \leq w(\OPT)/20$ for each element $e\in N$, is upper bounded by
\begin{align*}
	\frac{w(\OPT)}{20} \cdot \sum_{e \in \OPT} \frac{w(e)}{4} =
	\frac{w(\OPT)}{20} \cdot \frac{w(\OPT)}{4}
        =
	\frac{[w(\OPT)]^2}{80}
	.
\end{align*}

Combining these observations, we get by Chebyshev's inequality:
\begin{align*}
	&
	\Pr\left[w(\OPT \setminus S) < \frac{w(\OPT)}{7} \right]
	\leq
	\Pr\left[|w(\OPT \setminus S) - \mathbb{E}[w(\OPT \setminus S)]| > \left(\frac{1}{2} -
            \frac{1}{7}  \right)w(\OPT)\right]\\
	\leq{} &
	\Pr\left[|w(\OPT \setminus S) - \mathbb{E}[w(\OPT \setminus S)]| > \sqrt{10} \cdot \sqrt{\var[w(\OPT \setminus S)]}\right]
	\leq
	\frac{1}{10}
	.
\end{align*}

\end{description}

To summarize, we have by the union bound, $\Pr[G'] \geq 1 - (1/2 + 0.287 + 1/10) > 1/10$.
\end{proof}


\section{Missing Proofs of Section~\ref{sec:algo}}\label{sec:missingProofs}

\begin{proof}[Proof of Lemma~\ref{lem:composition}]
  For every $j\in H$, let $F_j = \bigcup_{i\in H, i\geq j} I_i$.  We show by induction that $F_j \in
  \mathcal{I}$ for all $j\in H$.  The result then follows by choosing $j$ to be the smallest index
  in $H$.  Clearly $F_j\in \cI$ when $j$ is the largest index in $H$, since in this case
  $F_j = I_j \in \cI_j \subseteq \cI$ (the last inclusion holds since $M_j$ is obtained from $M$ by restrictions and contractions).

Now assume that $j\in H$ is not the largest
index in $H$. Then, 
$F_j= F_{j + 2} \cup I_{j}$ since $H$ contains either odd
or even indices. By the induction
hypothesis we obtain $F_{j + 2} \in \mathcal{I}$.
Also, since
$I_i\subseteq N_i\subseteq \spn(S\cap B_{i-1})$,
we have
$F_{j + 2} \subseteq \spn(S\cap B_{\geq j+1})$.
The fact that $I_j\in \cI_j$, implies
\begin{align*}
|I_j| &= r_j(I_j) {=}
r(I_j \cup (S\cap B_{\geq j+1}))
- r(S\cap B_{\geq i+1})  &\mbox{(by the definition of $r_j$~\eqref{eq:ri})} \\
&= r(I_j \cup \spn(S\cap B_{\geq j+1}))
   - r(\spn(S\cap B_{\geq i+1}))\\
&\leq r(I_j \cup F_{j + 2}) - r(F_{j + 2}),
\end{align*}
where the equality on the second line follows since $r(A+e)=r(A)$ for any $e\in \spn(A)$,
and the inequality follows from the submodularity (diminishing returns) of $r$
and our previous observation that $F_{j + 2} \subseteq \spn(S\cap B_{\geq j+1})$.
It remains to observe that $r(F_{j + 2})=|F_{j + 2}|$ since
$F_{j + 2} \in \mathcal{I}$, and hence,
the above inequality implies $r(I_j\cup F_{j + 2})\geq |I_j|+|F_{j + 2}|$.
Thus, $F_j=F_{j + 2}\cup I_j\in \mathcal{I}$.
\end{proof}

\begin{proof}[Proof of Lemma~\ref{lem:metaSelectMany}]
Conditioned on $i\in H$, Algorithm~\ref{alg:BucketsSecretary}
selects elements of $N_i \setminus S$
greedily, and thus, the number of
elements selected from $B_i$ is $r(N_i\setminus S)$.
The key observation for lower bounding $r(N_i\setminus S)$ is that $N_i\cap \OPT$ is independent
in $M_i$, which we prove first. Since \OPT is a maximum
weight independent set, we have
\begin{equation}\label{eq:Bip1Span}
B_{\geq i+1} \subseteq \spn(\OPT \cap B_{\geq i+1}).
\end{equation}
This easily follows by recalling that the greedy algorithm produces
\OPT, and thus,
$r(\OPT\cap B_{\geq i+1})=r(B_{\geq i+1})$.
Hence,
\begin{align*}
r_i(N_i\cap \OPT) &=
  r((N_i\cap \OPT) \cup (S\cap B_{\geq i+1}))
    - r(S\cap B_{\geq i+1}) & \mbox{(by~\eqref{eq:ri})}\\
 &\geq
   r((N_i\cap \OPT)\cup (\OPT \cap B_{\geq i+1}))
     - r(\OPT\cap B_{\geq i+1}) & \mspace{-1mu} \mbox{(by~\eqref{eq:Bip1Span} and submodularity)}\\
 &= |(N_i\cap \OPT)\cup (\OPT\cap B_{\geq i+1})|
 - |\OPT\cap B_{\geq i+1}|\\
 &= |N_i\cap \OPT|,
\end{align*}
where the penultimate equality follows by observing
that $(N_i\cap \OPT)\cup (\OPT \cap B_{\geq i+1})$ is
independent in $M$ since it is a subset of \OPT.
Hence, we showed that $N_i\cap \OPT$ is independent
in $M_i$. Therefore,
$r(N_i\setminus S) \geq |N_i\cap \OPT\cap \bar{S}|$, where
$\bar{S}=N\setminus S$, and
we obtain
\begin{align*}
\E[|T\cap B_i|] &= \Pr[i\in H]\cdot \mathbb{E}[r(N_i\setminus S)]
    \geq \frac{1}{2}\cdot
    \E[|N_i \cap \OPT \cap \bar{S}|]\\
& =
\frac{1}{2} \cdot \sum_{e\in B_i\cap \OPT} \left[
\underbrace{\Pr[e\in \bar{S}]}_{=1/2}
\cdot \Pr[e\in N_i\mid e\not\in S]
\right].
\end{align*}
It remains to observe that for any element $e\in B_i$,
\begin{equation*}
\Pr[e\in N_i \mid e\not\in S] = p_{e,f(B_{i-1})}.
\end{equation*}
Indeed, if $i=1$ this follows from $N_1=B_1$
and $p_{e,0}=1$ (we recall that $f(B_0)=0$ by convention).
Otherwise, for $i>1$, we have
$N_i=B_i\cap \spn(S\cap B_{\geq i-1})$, and thus,
for any $e\in B_i$,
\begin{equation*}
\Pr[e\in N_i \mid e\not\in S]
  = \Pr[e\in \spn(S\cap B_{\geq i-1}) \mid e\not\in S]
  = p_{e,f(B_{i-1})}.
	\qedhere
\end{equation*}

\end{proof}

\end{document}